\documentclass[letterpaper, 10 pt, conference]{ieeeconf}  

\IEEEoverridecommandlockouts                              

\usepackage{graphicx}

\usepackage{color}
\usepackage{mathtools,amsmath,amssymb}
\usepackage{comment}
\usepackage[ruled,vlined]{algorithm2e} 
\usepackage{multirow, booktabs}

\usepackage{graphicx}
\usepackage{subfigure}
\usepackage{epsfig} 
\usepackage{cite}
\usepackage{bm}
\usepackage{bbm}
\usepackage{float}
\usepackage{url}
\usepackage{dsfont}

\newtheorem{theorem}{Theorem}[section]

\newtheorem{corollary}[theorem]{Corollary}
\newtheorem{definition}[theorem]{Definition}

\newtheorem{assumption}[theorem]{Assumption}
\newcommand{\blockend}{\hfill$\Box$}

\usepackage{tikz}

\usepackage{xcolor}   

\begin{document}

\twocolumn
\pagestyle{empty}
\setcounter{page}{1}

\title{Environment Parameter Gradient Theorem for Co-Design\\ in Reinforcement Learning}
\author{Amber Srivastava
\thanks{This work is partially supported by the Prime Minister Early Career Research Grant from the Anusandhan National Research Foundation (ANRF/ECRG/2024/004876) and the Faculty Seed Grant at IIT Delhi.}
\thanks{Amber Srivastava is with Faculty of Mechanical Engineering, Indian Institute of Technology Delhi, India.
        {\tt\small email: asrvstv@iitd.ac.in}}%
}


\maketitle


\thispagestyle{empty}
\pagestyle{empty}

\begin{abstract}
Reinforcement learning (RL) is traditionally concerned with learning a control policy for a fixed environment. In many engineering systems, however, the environment contains tunable physical or operational design parameters that shape the system's transition dynamics and costs experienced by the RL agent. This motivates jointly optimizing both the policy and the environment design parameters. To this end, we establish an Environment Parameter Gradient Theorem --- a formal expression for the gradient of the RL objective with respect to environment parameters. The key theoretical device is a generalized action-value function $Q_{\pi,\xi}(s,a,\zeta)$, which comprises two copies of the environment parameters: $\zeta$ governs the cost and transition dynamics at the current state--action pair, while $\xi$ governs the future rollouts. This decoupling yields a tractable closed-form gradient expression and is essential to the theorem's derivation. Building on this result, we develop a model-free algorithm that jointly optimizes the policy and environment design parameters. We demonstrate the efficacy of our framework on an {\color{black} aerial-relay communication network}, where the relay-node locations and message-routing policy are jointly optimized to minimize the total communication cost.
\end{abstract}


\section{Introduction} \label{sec: Introduction}

Reinforcement learning (RL) has emerged as a powerful framework for sequential decision-making under uncertainty \cite{Sutton1998}. In its standard formulation, RL focuses on learning a policy that controls a system while assuming that the environment itself remains fixed. However, in many areas such as robotics \cite{pmlr-v100-luck20a}, communication networks \cite{9517030}, building energy systems \cite{cauz2024reinforcement}, and wind turbines \cite{pao2024control}, there are design parameters that can also be optimized jointly with the agent's policy. {\color{black}Often these environment design parameters determine the underlying system dynamics and operational performance. Once optimized, these design parameters are fixed, and the agent continues to execute its optimal policy.}

{\color{black}For instance, consider the simplified communication relay network illustrated in Fig.~\ref{fig:UAV_Network}. Messages originate at fixed user nodes $n_i$ and must reach a destination node $\delta$ through a sequence of hops via aerial relay nodes $f_j$. The journey of a message is therefore a sequential decision process in which a routing policy selects the next hop at each node, while each transition incurs a communication cost that increases with the distance traveled. The objective is to learn an RL routing policy that routes messages through the network so as to minimize the total communication cost. However, the cost of each hop depends on the relay-node locations $\{y_j\}$, which determine the geometry of the network. Hence, improving the operational performance requires jointly optimizing the routing policy and the relay-node locations. We denote the collection of relay-node locations $\{y_j\}$ by the environment parameter $\zeta$. Once optimized, the aerial relay-node locations remain fixed in the network, while the optimal routing policy governs the message hops through the network.}

\begin{figure}
    \centering
    \includegraphics[width=\linewidth]{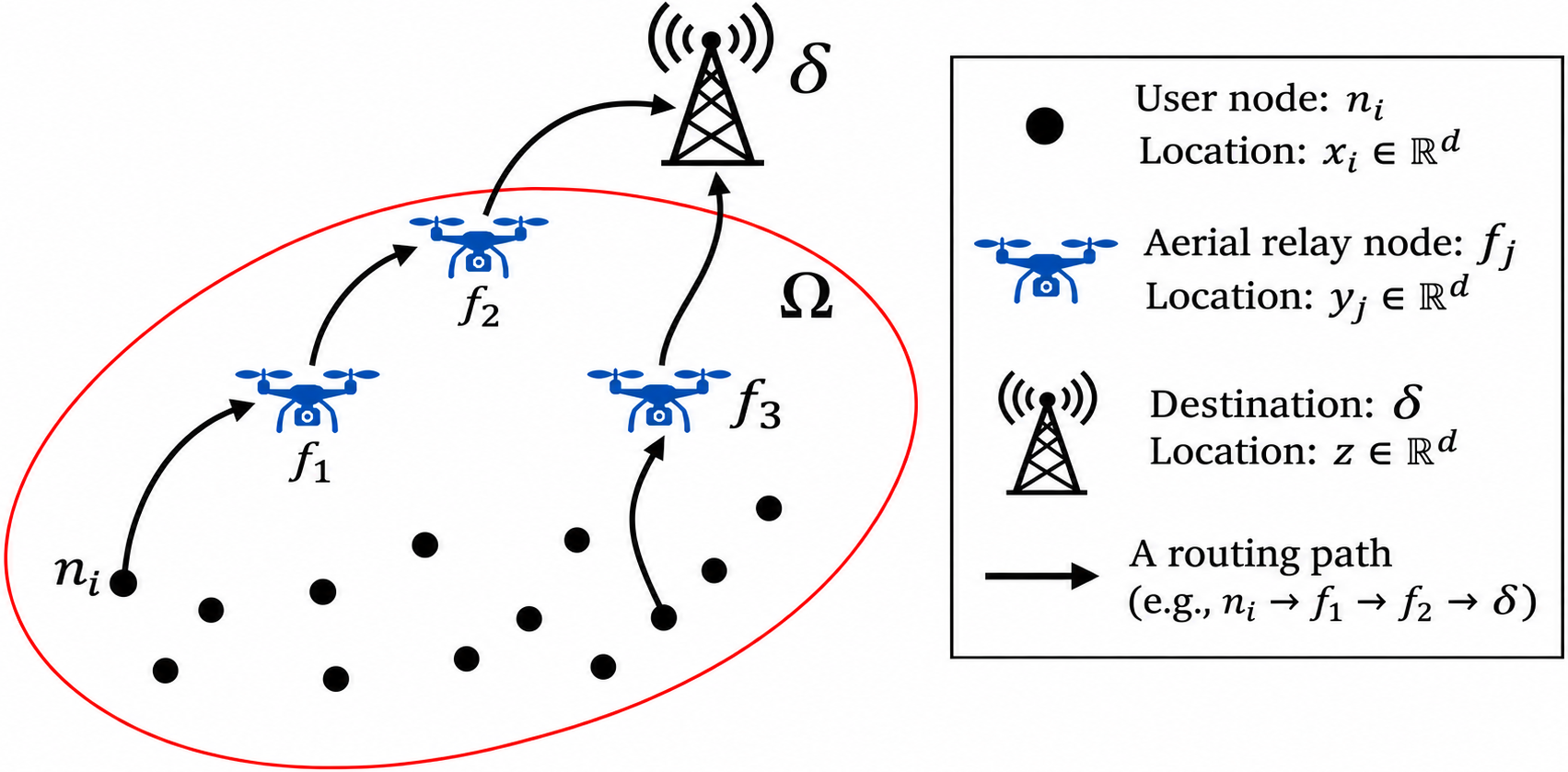} \vspace{-1em}
    \caption{{\color{black}Schematic of a communication relay network. Messages from user nodes $n_i$ are routed to the destination $\delta$ via aerial relay nodes $f_j$, whose locations $y_j$ form the environment parameter $\zeta$.}}
    \label{fig:UAV_Network} \vspace{-1.5em}
\end{figure}

These co-design problems have been addressed through several approaches in the literature. Classical co-design methods in control theory assume a parametric model of the system and jointly optimize the plant and controller in a model-based setting \cite{garcia2019control,pao2024control}. In contrast, RL-based methods have maintained a distribution over environment designs, adapting these distributions via likelihood-ratio (score-function) gradients alongside the policy, such as \cite{schaff2019jointly,cauz2024reinforcement,ha2019reinforcement}, while other approaches employ black-box optimization methods to search directly over the design space \cite{bhatia2021evolution,pmlr-v100-luck20a}. Note that the former estimate gradients of the objective with respect to the parameters of the distribution over the environment design parameters, rather than the gradients with respect to the design parameters themselves, while the latter do not exploit such gradient information at all. {\color{black}The parameterized-MDP framework in \cite{9517030} instead directly considers gradients with respect to the design parameters themselves; however, it assumes that the underlying system dynamics (transition kernel) are independent of these parameters. In contrast, our formulation allows the transition kernel to depend on $\zeta$, thereby accommodating a broader class of co-design problems.} Closer to our setting, \cite{jackson2021orchid, bolland2022jointly, he2024morph} directly compute gradients of the objective with respect to the design parameters; however, these approaches rely on access to a differentiable model of the underlying system dynamics.

An explicit expression of the gradient of the RL objective with respect to environment design parameters --- an analogue of the policy gradient theorem~\cite{sutton1999policy,silver2014deterministic} for these parameters --- would provide a principled foundation for gradient-based optimization, admitting model-free estimation from sampled trajectories, without relying on a distribution over designs or a known model of the system dynamics. To the best of our knowledge, no prior work provides such a tractable gradient expression that also admits model-free estimation from the trajectory rollouts.

Existing stochastic policy gradient (SPG) \cite{williams1992simple, sutton1999policy} and deterministic policy gradient (DPG) \cite{silver2014deterministic} theorems provide tractable gradient expressions for optimizing policy parameters $\theta$. {\color{black}A straightforward approach to leverage these results is to append the environment design parameter $\zeta$ to the action vector and work in an augmented action space. However, an unconstrained augmented policy could choose different values of the environment component at different states, whereas $\zeta$ in our setting is a state-independent design variable. Recovering the co-design problem would therefore require imposing structural restrictions on the admissible policy class over the augmented action space. Treating $\zeta$ separately is nevertheless nontrivial, since it determines the instantaneous cost and transition dynamics, while also conditioning the policy. Thus, $\zeta$ parameterizes both the underlying MDP and the policy that controls it, and the corresponding gradient must account for all of these dependencies.}

For gradient-based optimization of the RL objective $J(\pi,\zeta)$, we pursue two developments: (a) a tractable expression for the gradient with respect to $\zeta$, and (b) a model-free method to estimate it. To obtain (a), we introduce a generalized action-value function $\mathcal{Q}_{\pi,\xi}(s,a,\zeta)$ involving two copies of the environment parameter. This construction expresses the gradient as an expectation under the same trajectory distribution as $J(\pi,\zeta)$, which makes (b) possible.

Our main contributions are as follows. (i) We establish the Environment Parameter Gradient (EnvPG) Theorem for infinite-horizon MDPs. It gives a tractable expression for $\nabla_\zeta J(\pi,\zeta)$ that can be estimated from sampled trajectories without a model. {\color{black}It is the analogue, for the environment parameter, of the policy gradient theorem, and it holds even when the transition kernel itself depends on $\zeta$.} (ii) We introduce the two-copy generalized action-value function $\mathcal{Q}_{\pi,\xi}(s,a,\zeta)$, which enables the EnvPG expression. (iii) We develop a model-free actor--critic algorithm that learns the action-value function from sampled transitions (Section~\ref{sec: solution_methodology}) and jointly optimizes the policy and the environment design parameters. {\color{black}We demonstrate the EnvPG theorem and the corresponding algorithm on a synthetic aerial-relay communication network (Fig.~\ref{fig:UAV_Network}), where the relay locations and the message-routing policy are optimized jointly from sampled interaction, without a model of the network. Across the considered settings, EnvPG achieves objective values within about $1.0\%$ (on average) of a model-based oracle and about $3.4\%$ lower than a score-function co-design baseline.}

\section{Problem Formulation}\label{sec: ProbForm}
We consider a parameterized Markov decision process (MDP) $\mathcal{M}(\zeta) = \langle \mathcal{S}, \mathcal{A}, p_{\zeta}, c_{\zeta}, \gamma \rangle$, where $\mathcal{S}$ is a finite state space, $\mathcal{A}$ is a finite action space, $\zeta\in\mathbb{R}^d$ denotes the environment design parameter, $p_{\zeta}(s'|s,a)$ is the transition kernel giving the probability of transitioning to $s'$ from $s$ under action $a$, $c_{\zeta}(s,a,s') \geq 0$ is the instantaneous cost of the transition, and $\gamma \in (0,1]$ is the discount factor. Note that both the transition kernel and instantaneous cost depend on the environment parameter $\zeta$.

At each time step $t$, the agent observes the state $s_t$, selects action $a_t \in \mathcal{A}$ according to a stochastic policy $\pi(a_t|s_t,\zeta)$, transitions to the next state $s_{t+1} \sim p_{\zeta}(\cdot|s_t, a_t)$, and incurs cost $c_{\zeta}(s_t, a_t, s_{t+1})$. The policy $\pi(a|s,\zeta)$ is explicitly conditioned on $\zeta$, since the optimal action selection depends on the underlying MDP configuration. We denote by $p_{\pi,\zeta}$ the distribution over the {\em paths} $(\hdots,s_t,a_t,s_{t+1},a_{t+1},\hdots)$ of the MDP induced by policy $\pi$, and transition kernel $p_{\zeta}$. The  objective is to minimize the expected cumulative cost over both $\pi$ and $\zeta$ under an initial state distribution $\mu$ over $\mathcal{S}$:
\begin{align}\label{eq:objective}
\min_{\pi,\zeta}J(\pi, \zeta) := \sum_{s\in\mathcal{S}}\mu(s) V_{\pi,\zeta}(s),
\end{align}
where the state-value function $V_{\pi,\zeta}(s)$ is given by
\begin{align}
V_{\pi,\zeta}(s) = \mathbb{E}_{p_{\pi,\zeta}}\left[\sum_{t=0}^{\infty}
\gamma^t c_{\zeta}(s_t, a_t, s_{t+1}) \,\Big|\, s_0 = s\right].
\end{align}
The above value function satisfies the Bellman recursion
\begin{subequations}\label{eq:VQ}
\begin{align}
V_{\pi,\zeta}(s) &{=}\sum_{a\in\mathcal{A}} \pi(a|s,\zeta)\,Q_{\pi,\zeta}(s,a), \label{eq:VinQ}\\
Q_{\pi,\zeta}(s,a) &{=}\sum_{s'\in\mathcal{S}} p_{\zeta}(s'|s,a)\Big(c_{\zeta}(s,a,s') {+} \gamma V_{\pi,\zeta}(s')\Big), \label{eq:OriginalQ}
\end{align}
\end{subequations}
where $Q_{\pi,\zeta}(s,a)$ is referred to as the action-value function. 

{\color{black}
\begin{assumption}\label{asmp:wellposed}
The instantaneous costs are bounded. If $\gamma=1$, the problem is additionally a stochastic shortest-path problem with a cost-free absorbing terminal state, and the admissible policies are proper for every $\zeta$ considered.\hfill\blockend
\end{assumption}

Under this assumption $V_{\pi,\zeta}(s)$ and $J(\pi,\zeta)$ are finite \cite[Chapter~2]{bertsekas1996neuro}.} For a fixed $\zeta$, the minimization over $\pi$ in~\eqref{eq:objective} reduces to a standard RL problem \cite{Sutton1998}. The key challenge addressed in this work is the minimization over $\zeta$.

\section{Environment Parameter Gradient {\color{black}(EnvPG)}}\label{sec: ProbSoln}
A key property of the policy gradient theorem \cite{sutton1999policy} is that the gradient is an expectation under the distribution induced by the policy itself and contains no derivative of the transition kernel, so it can be estimated from sampled trajectories without knowing the kernel. We seek a similar property for the environment parameter $\zeta$, that is, an expression for $\nabla_\zeta J(\pi,\zeta)$ that is an expectation under the distribution induced by the policy, with no terms involving the transition kernel, its gradient, or the gradient of the instantaneous cost, and is therefore amenable to model-free estimation. This is enabled through a generalized action-value function that formally separates the environment parameter into two copies, each playing a distinct role.

\begin{definition}\label{def:GenQ_func}
For a policy $\pi$, the generalized action-value function $\mathcal{Q}_{\pi,\xi}(s,a,\zeta)$ is defined as:
\begin{align}\label{eq:generalized_Q}
\mathcal{Q}_{\pi,\xi}(s,a,\zeta) {=} \sum_{s'} p_{\zeta}(s'|s,a)\Big(c_{\zeta}(s,a,s') {+} \gamma V_{\pi,\xi}(s')\Big),
\end{align}
where the copy $\zeta$, appearing as the third argument alongside $s,a$, is the {\em current-step parameter copy} that determines the transition kernel $p_\zeta(\cdot|s,a)$ and cost $c_\zeta(s,a,\cdot)$ at the $(s,a)$ pair, while the copy $\xi$, appearing in the subscript alongside $\pi$, is the {\em future-rollout parameter copy} that governs the environment over all
subsequent steps starting at $s'$. \blockend
\end{definition}

Note that when $\zeta = \xi$, the generalized $\mathcal{Q}_{\pi,\xi}(s,a,\zeta)$ (or, equivalently $Q_{\pi,\zeta}(s,a,\zeta)$) in \eqref{eq:generalized_Q} reduces to the standard action-value function $Q_{\pi,\zeta}(s,a)$ in~\eqref{eq:OriginalQ}. Thus, the state-value function $V_{\pi,\zeta}(s)$ in \eqref{eq:VinQ} is conveniently expressible in terms of the generalized action-value function:
\begin{align}\label{eq:VinGenQ}
V_{\pi,\zeta}(s) = \sum_{a\in\mathcal{A}}\pi(a|s,\zeta)\mathcal{Q}_{\pi,\zeta}(s,a,\zeta).
\end{align}

\begin{theorem}[Environment Parameter Gradient - EnvPG]\label{thm:epgt}
Let $\pi(\cdot|s,\zeta)$ be a differentiable stochastic policy and $\mathcal{Q}_{\pi,\xi}(s,a,\zeta)$ be the generalized action-value function in \eqref{eq:generalized_Q}. The gradient of the objective $J(\pi,\zeta)$ with respect to the environment parameter $\zeta$, evaluated at $\zeta = \bar{\zeta}$, is given by:
\begin{align}\label{eq:theorem_J}
&\nabla_\zeta J(\pi,\zeta)\Big|_{\zeta=\bar{\zeta}} =\mathbb{E}_{p_{\pi,\bar{\zeta}}}\Bigg[\sum_{t=0}^{\infty}\gamma^t 
\nabla_\zeta\Big(\mathcal{Q}_{\pi,\bar{\zeta}}(s_t,a_t,\zeta) + \nonumber \\
&\qquad \quad\ln\pi(a_t|s_t,\zeta) \cdot Q_{\pi,\bar{\zeta}}(s_t,a_t,\bar{\zeta})\Big)
\Big|_{\zeta=\bar{\zeta}}\,\Big|\,s_0 \sim \mu\Bigg].
\end{align}
\end{theorem}
\begin{proof}
\noindent Differentiating state-value function $V_{\pi,\zeta}(s,a)$ in \eqref{eq:VinGenQ} with respect to $\zeta$ and evaluating at $\zeta = \bar{\zeta}$:
\begin{align}\label{eq:gradVZetaFirst}
\nabla_\zeta V_{\pi,\zeta}(s)\Big|_{\zeta=\bar{\zeta}} 
&= \sum_a \nabla_\zeta \pi(a|s,\zeta)\big|_{\zeta=\bar{\zeta}} 
\mathcal{Q}_{\pi,\bar{\zeta}}(s,a,\bar{\zeta}) \nonumber\\
&\quad + \sum_a \pi(a|s,\bar{\zeta}) \nabla_\zeta 
\mathcal{Q}_{\pi,\zeta}(s,a,\zeta)\big|_{\zeta=\bar{\zeta}}.
\end{align}
For the second term in (\ref{eq:gradVZetaFirst}), we expand $\nabla_\zeta \mathcal{Q}_{\pi,\zeta}(s,a,\zeta)\big|_{\zeta=\bar{\zeta}}$ using its definition in \eqref{eq:generalized_Q} and the  chain rule that result into
\begin{align}\label{eq:GenQGrad_zeta}
\nabla_\zeta \mathcal{Q}_{\pi,\zeta}(s,a,& \zeta)\big|_{\zeta=\bar{\zeta}} 
= \nabla_\zeta \mathcal{Q}_{\pi,\bar{\zeta}}(s,a,\zeta)\big|_{\zeta=\bar{\zeta}} \nonumber\\
&\qquad + \gamma\sum_{s'} p_{\bar{\zeta}}(s'|s,a) \nabla_\zeta V_{\pi,\zeta}(s')\big|_{\zeta=\bar{\zeta}},
\end{align}
Substituting \eqref{eq:GenQGrad_zeta} back into \eqref{eq:gradVZetaFirst}, rewriting $\nabla_\zeta \pi(a|s,\zeta)$ in (\ref{eq:gradVZetaFirst}) as $\pi(a|s,\zeta)\nabla_\zeta\ln \pi(a|s,\zeta)$, and collecting the common factor $\pi(a|s,\bar{\zeta})$ we obtain
\begin{align}
\nabla_\zeta V_{\pi,\zeta}(&s)\Big|_{\zeta=\bar{\zeta}} = \sum_a \pi(a|s,\bar{\zeta})\Big[ \nabla_\zeta\Big(\mathcal{Q}_{\pi,\bar{\zeta}}(s,a,\zeta)\nonumber\\
&\qquad \quad ~ + \ln\pi(a|s,\zeta)\mathcal{Q}_{\pi,\bar{\zeta}}(s,a,\bar{\zeta})\Big)\Big|_{\zeta=\bar{\zeta}}\nonumber\\
&\qquad + \gamma\sum_{s'} p_{\bar{\zeta}}(s'|s,a) 
\nabla_\zeta V_{\pi,\zeta}(s')\Big|_{\zeta=\bar{\zeta}}\Big].
\end{align}
This recursion has the same structure as the ordinary Bellman equation for $V_{\pi,\zeta}(s)$ in \eqref{eq:VQ}, where the bracketed term involving $\mathcal{Q}$ plays the role of the instantaneous cost. Such a recursion is solved by the expectation of these instantaneous cost accumulated along trajectories generated under $\pi$ and $p_{\bar\zeta}$, i.e., under the trajectory distribution $p_{\pi,\bar\zeta}$. This gives
\begin{align}\label{eq:theorem}
&\nabla_\zeta V_{\pi,\zeta}(s)\Big|_{\zeta=\bar{\zeta}} = \mathbb{E}_{p_{\pi,\bar{\zeta}}}\Bigg[\sum_{t=0}^{\infty}\gamma^t 
\nabla_\zeta\Big(\mathcal{Q}_{\pi,\bar{\zeta}}(s_t,a_t,\zeta) + \nonumber\\ 
&\qquad\quad\ln\pi(a_t|s_t,\zeta) \cdot \mathcal{Q}_{\pi,\bar{\zeta}}(s_t,a_t,\bar{\zeta})\Big)\Big|_{\zeta=\bar{\zeta}}\,\Big|\,s_0 = s\Bigg],
\end{align}
The gradient of $J(\pi,\zeta)$ in \eqref{eq:theorem_J} follows immediately by taking the expectation of $\nabla_\zeta V_{\pi,\zeta}(s)\big|_{\zeta=\bar{\zeta}}$ over $s\sim\mu$.
\end{proof}

{\color{black}\begin{assumption}\label{asmp:gradient}
The kernel $p_\zeta$ and the cost $c_\zeta$ are continuously differentiable in $\zeta$. The policy $\pi(a|s,\zeta)$ is differentiable in $\zeta$ and, at every $\bar{\zeta}$ considered, satisfies $\pi(a|s,\bar{\zeta})>0$ with $\nabla_\zeta\log\pi(a|s,\zeta)|_{\zeta=\bar{\zeta}}$ finite for all $(s,a)$.\hfill\blockend
\end{assumption}

Under Assumptions \ref{asmp:wellposed} and \ref{asmp:gradient}, and since $\mathcal{S}$ and $\mathcal{A}$ are finite, the one-step gradient (the summand of~\eqref{eq:theorem}) is bounded, so the expectation in \eqref{eq:theorem} is well-defined~\cite{bertsekas1996neuro}. Note that using the standard state-visitation representation \cite[Chapter 13]{Sutton1998}, \eqref{eq:theorem_J} can equivalently be written as
\begin{align}\label{eq:theorem_w}
\nabla_\zeta J(\pi,\zeta)\Big|_{\zeta=\bar{\zeta}}
&=\sum_{s}w_{\pi,\bar{\zeta}}(s)\sum_{a}\pi(a|s,\bar{\zeta})\nabla_\zeta\Big[
\mathcal{Q}_{\pi,\bar{\zeta}}(s,a,\zeta)\nonumber\\
&\quad+\ln\pi(a|s,\zeta)\,\mathcal{Q}_{\pi,\bar{\zeta}}(s,a,\bar{\zeta})\Big]\Big|_{\zeta=\bar{\zeta}},
\end{align}
where $w_{\pi,\bar\zeta}(s)$ is the discounted visitation weight defined as 
\begin{align}\label{eq:w_visit_weight}
w_{\pi,\bar{\zeta}}(s)=\sum_{t=0}^{\infty}\gamma^{t}\,\Pr(s_t=s\,|\,\pi,\bar{\zeta},s_0\sim \mu),
\end{align}
for all non-absorbing states $s\in\mathcal{S}$, and $0$ otherwise.} Both terms in the summand of~\eqref{eq:theorem_w} can be estimated from sampled trajectories without a model of the environment, using a critic that estimates $\mathcal{Q}_{\pi,\xi}(s,a,\zeta)$ and an actor that represents $\pi(a|s,\zeta)$, as discussed in Section~\ref{sec: solution_methodology}.

\section{\color{black}{Entropy-Regularized EnvPG Actor--Critic}}\label{sec: solution_methodology}
To encourage exploration and improve training, we augment the  objective \eqref{eq:objective} with a common entropy regularization term \cite{haarnoja2018soft,Fox2016Glearning}, yielding the entropy-regularized objective:
\begin{align}\label{eq:entropy_objective}
J_{\mathrm{H}}(\pi,\zeta){=}\mathbb{E}_{p_{\pi,\zeta}}\Big[\sum_{t=0}^{\infty}\hspace{-0.1cm}\gamma^t\big(c_{t,\zeta}{+}\alpha\ln\pi(a_t|s_t,\zeta)\big){\big|}s_0{\sim}\mu\Big]
\end{align}
where $c_{t,\zeta} = c_{\zeta}(s_t,a_t,s_{t+1})$, and $\alpha > 0$ is a temperature parameter. Applying the EnvPG Theorem \ref{thm:epgt} to the regularized objective (\ref{eq:entropy_objective}) and {\color{black} re-writing it in terms of the discounted visitation weights $w_{\pi,\bar\zeta}$ as in (\ref{eq:theorem_w}) yields:
\begin{align}\label{eq:entropy_envpg_w}
&\nabla_\zeta J_{\mathrm{H}}(\pi,\zeta)\Big|_{\zeta=\bar{\zeta}}
=\sum_{s}w_{\pi,\bar{\zeta}}(s)\sum_{a}\pi(a|s,\bar{\zeta})\nonumber\\
&\quad\nabla_\zeta\Big[
\mathcal{Q}^{\mathrm{H}}_{\pi,\bar{\zeta}}(s,a,\zeta)+\ln\pi(a|s,\zeta)\,\mathcal{Q}^{\mathrm{H}}_{\pi,\bar{\zeta}}(s,a,\bar{\zeta})\Big]\Big|_{\zeta=\bar{\zeta}},
\end{align}}
\noindent where the generalized entropy-regularized action-value function $\mathcal{Q}^{\mathrm{H}}_{\pi,\xi}(s,a,\zeta)$ is defined analogously to Definition \ref{def:GenQ_func}:
\begin{align}\label{eq:entReg_Q}
\mathcal{Q}^{\mathrm{H}}_{\pi,\xi}(s,a,\zeta){=}\hspace{-0.12cm}\sum_{s'} p_{\zeta}(s'|s,a)\big(c^{\mathrm{H}}_{\zeta}(s,a,s'){+} \gamma V_{\pi,\xi}^{\mathrm{H}}(s')\big),
\end{align}
where $c^{\mathrm{H}}_{\zeta}(s,a,s')=c_{\zeta}(s,a,s')+\alpha \ln\pi(a|s,\zeta)$, and the current-step and the future rollout parameter copies are denoted by $\zeta$ and $\xi$, respectively. Here, $V^{\mathrm{H}}_{\pi,\xi}(s')$ denotes the entropy regularized state value function. 

Let $\pi_\theta$ denote a policy parameterized by $\theta$. The gradient of $J_{\mathrm{H}}(\pi_\theta,\zeta)$ with respect to $\theta$ follows from the standard entropy-regularized policy gradient theorem \cite{haarnoja2018soft}:
{\color{black}\begin{align}\label{eq:entropy_spg_w}
\nabla_{\theta}J_{\mathrm{H}}(\pi_\theta,\zeta)
&=\sum_{s}w_{\pi_\theta,\zeta}(s)\sum_{a}\pi_\theta(a|s,\zeta)\nonumber\\ 
&~ \big(\mathcal{Q}^{\mathrm{H}}_{\pi_\theta,\zeta}(s,a,\zeta)+\alpha\big)\,\nabla_\theta\log\pi_\theta(a|s,\zeta).
\end{align}}
where the additive $\alpha$ arises from the entropy term in (\ref{eq:entropy_objective}).

Under {\em direct parameterization} of the policy, i.e., a separate policy parameter $\theta_{s,a}^{\zeta}$ ($=:\pi(a|s,\zeta)$) for each state--action pair, the gradient expression of Theorem~\ref{thm:epgt} reduces to a single term. Here, the optimal policy for the entropy-regularized objective~\eqref{eq:entropy_objective} is known to be the Gibbs distribution \cite{9517030,Fox2016Glearning}:
\begin{align}\label{eq:gibbs_policy}
\pi^*(a|s,\zeta) = \frac{\exp\{-\frac{1}{\alpha}\mathcal{Q}^{\mathrm{H}}_{\pi^*,\zeta}(s,a,\zeta)\}}
{\sum_{a'\in\mathcal{A}}\exp\{-\frac{1}{\alpha}\mathcal{Q}^{\mathrm{H}}_{\pi^*,\zeta}(s,a',\zeta)\}},
\end{align}
which upon substitution in the value function $V_{\pi^*,\zeta}^{\mathrm{H}}(s)=\sum_a \pi^*(a|s,\zeta)Q_{\pi^*,\zeta}^{\mathrm{H}}(s,a,\zeta)$ yields the soft-min expression:
\begin{align}\label{eq:softmin}
V^{\mathrm{H}}_{\pi^*,\zeta}(s) = -\alpha\log\bigg({\sum_{a\in\mathcal{A}}}\hspace{-0.1cm}\exp\bigg({-}\frac{\mathcal{Q}^{\mathrm{H}}_{\pi^*,\zeta}(s,a,\zeta)}{\alpha}\bigg)\bigg).
\end{align}

Subsequently, the gradient in (\ref{eq:entropy_envpg_w}) simplifies as follows.
\begin{corollary}\label{cor:gibbs}
For the objective $J_{\mathrm{H}}(\pi,\zeta)$ in (\ref{eq:entropy_objective}), with optimal policy $\pi^*$ in (\ref{eq:gibbs_policy}), the gradient $\nabla_\zeta J_{\mathrm{H}}(\pi^*,\zeta)|_{\zeta=\bar{\zeta}} = $
\begin{align}\label{eq:gibbs_gradient}
\sum_{s\in\mathcal{S}} w_{\pi^*,\bar\zeta}(s) \sum_{a\in\mathcal{A}} \pi^*(a|s,\bar{\zeta})\,
\nabla_\zeta \mathcal{Q}^{\mathrm{H}}_{\pi^*,\bar{\zeta}}(s,a,\zeta)\Big|_{\zeta=\bar{\zeta}},
\end{align}
\end{corollary}

\begin{proof}
Differentiating the soft-min value function \eqref{eq:softmin} with respect to $\zeta$ and evaluating at $\zeta=\bar{\zeta}$:
\begin{align}\label{eq:EntRegVal_s_Grad}
\hspace{-0.5cm}\nabla_\zeta V^{\mathrm{H}}_{\pi^*,\zeta}(s)\big|_{\zeta=\bar{\zeta}} {=} 
{\sum_{a\in\mathcal{A}}}\pi^*(a|s,\bar{\zeta})
\nabla_\zeta \mathcal{Q}^{\mathrm{H}}_{\pi^*,\zeta}(s,a,\zeta)\big|_{\zeta=\bar{\zeta}}, \hspace{-0.5cm}
\end{align}
where $\nabla_\zeta \mathcal{Q}^{\mathrm{H}}_{\pi^*,\zeta}(s,a,\zeta)\big|_{\zeta=\bar{\zeta}}$ can be evaluated in a manner analogous to the $\nabla_\zeta \mathcal{Q}_{\pi,\zeta}(s,a,\zeta)\big|_{\zeta=\bar{\zeta}}$ in (\ref{eq:GenQGrad_zeta}). When substituted back in (\ref{eq:EntRegVal_s_Grad}), and re-written in terms of the state visitation weight $w_{\pi^*,\bar\zeta}(s)$ we obtain the expression (\ref{eq:gibbs_gradient}).
\end{proof}

\begin{algorithm}[t]
\caption{EnvPG Actor-Critic Algorithm}
\label{alg:envpg_ac}
\textbf{Input:} $\mathcal{S},\mathcal{A}$, $\gamma$, actor network $f_\theta$, critic network $\mathcal{Q}_\phi$, step sizes $\{\iota_k\}$, $\{\beta_k\}$, $\{\eta_k\}$, temperature $\alpha$, annealing rate $\upsilon<1$, episodes per iteration $M$, max outer iteration $k_{\max}$, variance $\sigma^2$, batch size $B$, {\color{black}trajectory counts for stochastic estimates $N_1$, $N_2$}\\
\textbf{Initialize:} $\theta_0$, $\phi_0$, $\phi^- \leftarrow \phi_0$, $\zeta_0$, $k \leftarrow 0$\\
\While{$k\leq k_{\max}$}{
    Replay buffer $\mathcal{D}_k \leftarrow \emptyset$\\
    \For{episode $= 1$ \textbf{to} $M$}{
        Sample initial state $s_0$, $t \leftarrow 0$\\
        \While{until termination}{
            $\bm{\epsilon}_{\mathrm{nb}} \sim \mathcal{N}(0, \sigma^2 \mathbf{I})$; $\zeta_{\text{pert}} \leftarrow \zeta_k + \bm{\epsilon}_{\text{nb}}$\\
            {\color{black}sample $a_t \sim 
            \pi_{\theta_k}(\cdot | s_t,\zeta_{\text{pert}})$}\\
            take action $a_t$ with environment at $\zeta_{\text{pert}}$\\ 
            observe $s_{t+1}$ and $c_t$\\
            $\mathcal{D}_k {\leftarrow }\mathcal{D}_k {\cup} \{(s_t,a_t,s_{t+1},c_t,\zeta_{\text{pert}})\}$\\
            $t \leftarrow t+1$\\
        }
    }
    sample $\{(s^{(i)},a^{(i)},s'^{(i)},c^{(i)},\zeta^{(i)})\}_{i=1}^B$ from $\mathcal{D}_k$\\
    compute TD target $y^{(i)}$ for all $i$ in \eqref{eq:td_target}\\
    {\color{black}update critic $\phi_{k+1}\leftarrow \phi_k - \iota_k \nabla_\phi\mathcal{L}(\phi_k)$ using (\ref{eq:critic_loss})}\\
    update target $\phi^{-}$ via soft Polyak updates \cite{Sutton1998}\\
{\color{black}sample $N_1$ trajectories at $\zeta_k$ $\pi_{\theta_k}$ for $\hat{h}_{\pi_{\theta_k}}^{\zeta_k}$ in \eqref{eq:grad_est_theta}}\\    
    {\color{black}update actor: $\theta_{k+1}\leftarrow \theta_k - \beta_k\,\hat{h}_{\pi_{\theta_k}}^{\zeta_k}$}\\
{\color{black}sample $N_2$ trajectories at $\zeta_k$, $\pi_{\theta_{k}}$ for $\hat{g}_{\pi_{\theta_{k}}}^{\zeta_k}$ in \eqref{eq:grad_est_zeta}}\\
{\color{black}update design: $\zeta_{k+1}\leftarrow \zeta_k - \eta_k\,\hat{g}^{\zeta_k}_{\pi_{\theta_{k}}}$}\\
$k \leftarrow k+1$; $\alpha\leftarrow \upsilon \alpha$
}
\textbf{Output:} $(\theta_{k_{\max}}, \zeta_{k_{\max}})$
\end{algorithm}

{\color{black}We now develop a model-free algorithm that jointly optimizes the policy and environment using the gradients in \eqref{eq:entropy_envpg_w} and \eqref{eq:entropy_spg_w}. In practice, these gradients are replaced by stochastic estimates obtained from sampled trajectories. Let $\pi_\theta(a|s,\zeta)$ denote a policy parameterized by an actor network $f_\theta(s,a,\zeta)$ with weights $\theta$, and $\mathcal{T}=\{\tau^i\}_{i=1}^{N}$ be $N$ trajectories $\tau^i=(s_0^i,a_0^i,s_1^i,a_1^i,\dots)$ collected under parameter $\bar\zeta$ and policy $\pi_{\theta}$ with $s_0^i\sim \mu$. We define
\begin{align}\label{eq:w_hat}
\hat{w}(s) = \frac{1}{N}\sum_{i=1}^N \sum_t\gamma^t \mathds{1}\{s_t^i = s\}
\end{align}
as the empirical estimate of state-visitation weight $w_{\pi_{\theta},\bar\zeta}(s)$ for all non-absorbing states $s\in\mathcal{S}$, and $\hat{w}(s)=0$ for all absorbing states. Here $\mathds{1}\{s_t^i=s\}$ denotes the indicator function, and $\mathbb{E}(\hat{w}(s)) = w_{\pi_\theta,\bar\zeta}(s)$. The unbiased stochastic estimate of \eqref{eq:entropy_envpg_w} is given by

\begin{align}\label{eq:grad_est_zeta}
\hat g^{\bar\zeta}_{\pi_\theta}&=\sum_{s}\hat w(s)\sum_{a}\pi_\theta(a|s,\bar\zeta)\Big[
\nabla_\zeta\mathcal{Q}^{\mathrm{H}}_{\pi_\theta,\bar\zeta}(s,a,\zeta)\big|_{\zeta=\bar\zeta}\nonumber\\
&\quad\qquad +\nabla_\zeta\ln\pi_\theta(a|s,\zeta)\big|_{\zeta=\bar\zeta}\,\mathcal{Q}^{\mathrm{H}}_{\pi_\theta,\bar\zeta}(s,a,\bar\zeta)\Big],
\end{align}
and the unbiased stochastic estimate of \eqref{eq:entropy_spg_w} is given by
\begin{align}\label{eq:grad_est_theta}
\hat h^{\bar\zeta}_{\pi_\theta}&=\sum_{s}\hat w(s)\sum_{a}\pi_\theta(a|s,\bar\zeta)\nonumber \\
&\qquad\big(\mathcal{Q}^{\mathrm{H}}_{\pi_\theta,\bar\zeta}(s,a,\bar\zeta)+\alpha\big)\nabla_\theta\ln\pi_\theta(a|s,\bar\zeta).
\end{align}}
Let $\mathcal{Q}_\phi(s,a,\zeta)$ denote the critic network that approximates the generalized action-value function in \eqref{eq:entReg_Q}. The gradient estimate \eqref{eq:grad_est_zeta} requires capturing how $\mathcal{Q}^{\mathrm{H}}_{\pi,\bar{\zeta}}(s,a,\zeta)$ varies with the current-step parameter copy $\zeta$, in a neighborhood of $\bar{\zeta}$. It does not require variation with respect to the future rollout copy; there, only the value at $\xi = \bar{\zeta}$ is needed. This asymmetry guides the learning of the critic $\mathcal{Q}_\phi(s,a,\zeta)$, that approximates $\mathcal{Q}_{\pi,\bar{\zeta}}^{\mathrm{H}}(s,a,\zeta)$. In particular, the critic's dependence on the current-step copy $\zeta$ is an explicit network input, while the dependence on the future rollout copy $\xi$ (which is fixed at $\bar{\zeta}$) and the policy $\pi$ is implicit.

\begin{figure*}[t]
\centering
\setlength{\tabcolsep}{0pt}
\begin{tabular*}{\textwidth}{@{\extracolsep{\fill}}ccccc@{}}
\includegraphics[height=3.0cm]{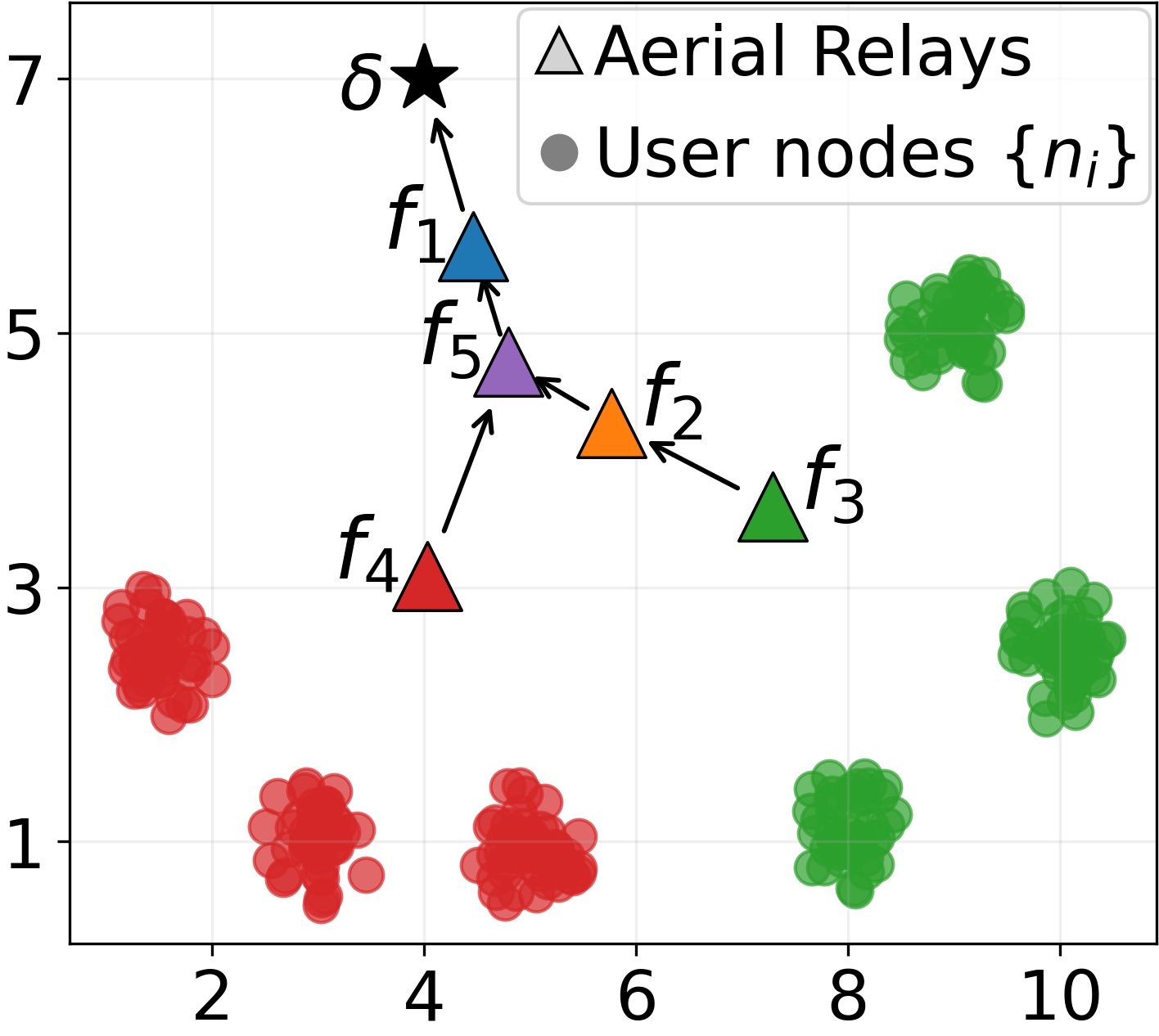} & 
\includegraphics[height=3.0cm]{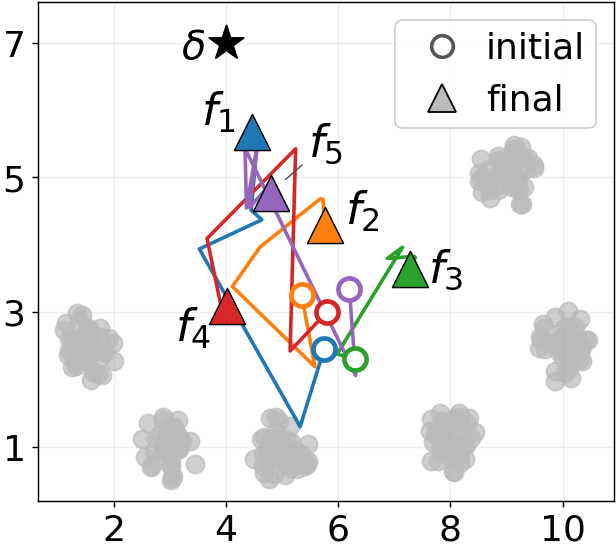} &
\includegraphics[height=3.0cm]{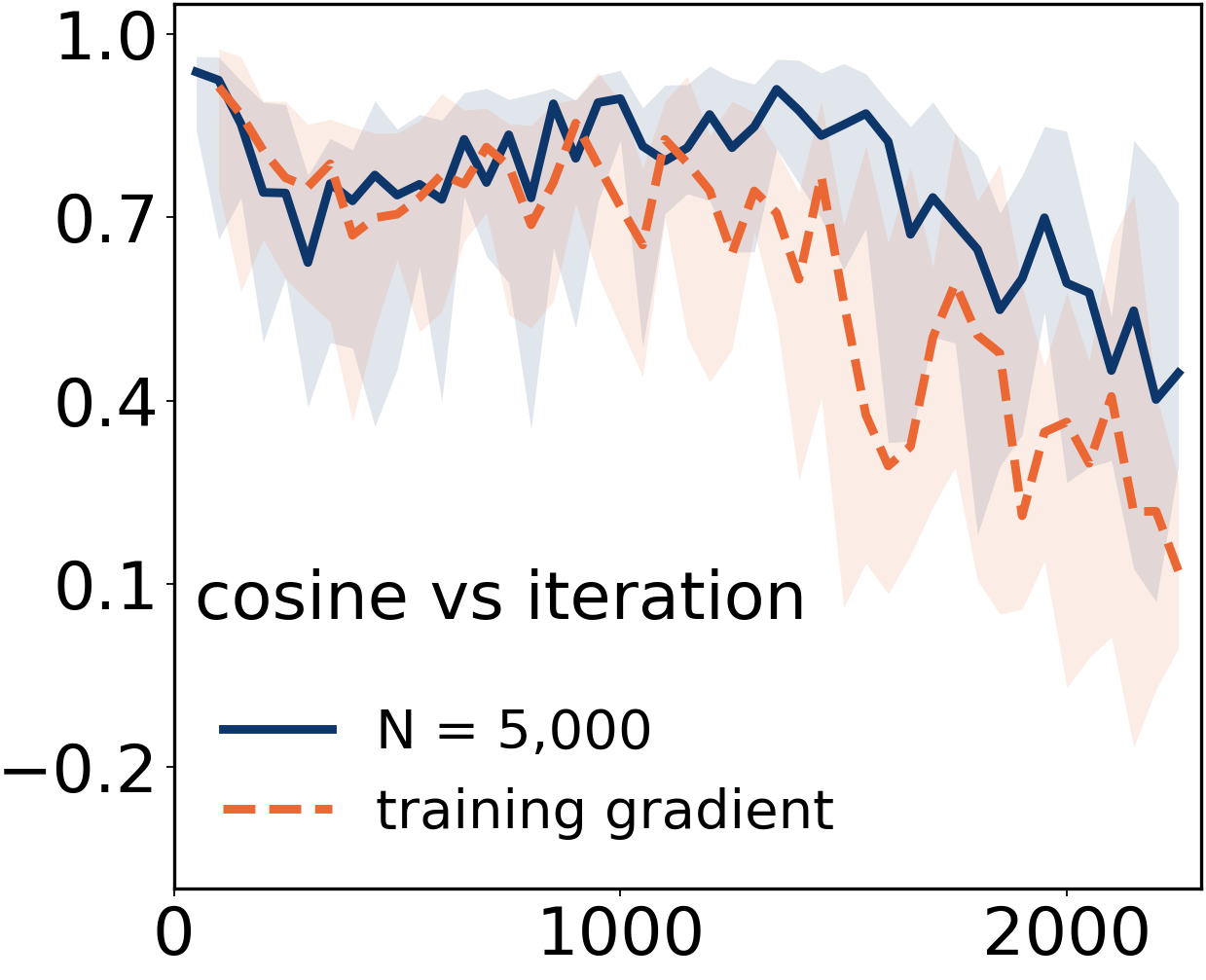} &
\includegraphics[height=3.0cm]{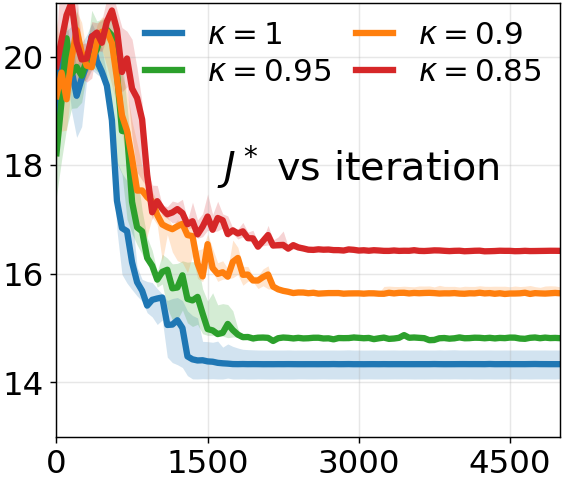} & 
\includegraphics[height=3.0cm]{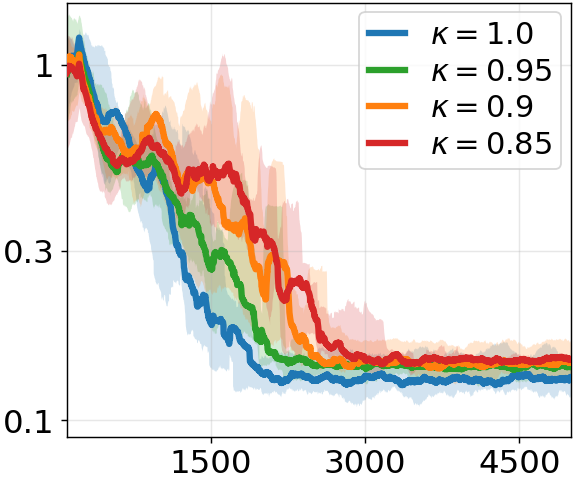}\\[-2pt]
{\footnotesize (a) Learned design, $\kappa=0.90$} &
{\footnotesize (b) Trajectories, $\kappa=0.90$} &
{\footnotesize (c) Cosine similarity} &
{\footnotesize (d) $J^*(\zeta_k)$ vs iteration} &
{\footnotesize (e) $\|\nabla_\zeta J_{\mathrm{H}}\|$ vs iteration}\\
\end{tabular*} \vspace{-0.3cm}
\caption{{\color{black}Aerial relay network design: (a) relay placement and routing learned by Algorithm~\ref{alg:envpg_ac} for $\kappa=0.90$, with user nodes colored by their assigned relay and arrows indicating the routing under the learned policy; (b) relay positions every $500$ updates for the same run, joined by straight lines, from the initial (circles) to the final (triangles) positions; (c) cosine similarity between the exact gradient and the stochastic gradient estimate obtained from $N=5{,}000$ trajectories (solid) and the stochastic gradient estimate (from $51$ sampled trajectories) observed in training (dashed) (median over all $24$ runs, shaded: interquartile range); (d) exact optimal cost $J^*(\zeta_k)$ for each $\kappa$ (median over initializations, shaded: interquartile range); (e) norm of the stochastic gradient estimate $\|\nabla_\zeta J_{\mathrm{H}}\|$ for each $\kappa$ (median over initializations of its moving average over $200$ updates, shaded: range over initializations).}}
\label{fig:network_designs}
\vspace{-0.5cm}
\end{figure*}

We train the critic by temporal-difference (TD) learning, building up on the methodology of deep Q-networks (DQN)~\cite{mnih2015human}. At the $k$-th iteration, with the environment parameter value $\zeta_k$ and the actor weight value $\theta_k$ we collect $B$ data points of the form $(s^{(i)}, a^{(i)}, s'^{(i)}, c^{(i)}, \zeta_{\text{pert}}^{(i)})$, with $a^{(i)} \sim \pi_{\theta_k}(\cdot|s^{(i)},\zeta_{\text{pert}}^{(i)})$, where $s'^{(i)}$ is the successor state reached when action $a^{(i)}$ is taken in state $s^{(i)}$, $c^{(i)}$ is the transition cost, $\zeta_{\text{pert}}^{(i)} = \zeta_k + \bm{\epsilon}_{\mathrm{nb}}$ is the perturbed environment parameter at which the transition is observed, and $\bm{\epsilon}_{\mathrm{nb}}\sim \mathcal{N}(0,\sigma^2\mathbf{I})$. For each data point, the TD target is
\begin{align}\label{eq:td_target}
y^{(i)} &= c^{(i)} + \alpha \log \pi_{\theta_k}(a^{(i)}|s^{(i)},\zeta_{\text{pert}}^{(i)})\nonumber \\
&\qquad + \gamma \sum_{a'\in\mathcal{A}} \pi_{\theta_k}(a'|s'^{(i)},\zeta_k)\,\mathcal{Q}_{\phi^-}(s'^{(i)},a',\zeta_k),
\end{align}
where $\mathcal{Q}_{\phi^-}$ is the target critic, a slowly-updated critic copy that stabilizes TD learning \cite{mnih2015human}. The first two terms of $y^{(i)}$ depend on the perturbed parameter $\zeta_{\text{pert}}^{(i)}$, while the last term depends on the baseline $\zeta_k$. This maps $\zeta_{\text{pert}}^{(i)}$ to the role of the current-step parameter copy and $\zeta_k$ to the role of the future rollout parameter copy. The critic parameters $\phi$ are obtained by minimizing the mean-squared error \cite{mnih2015human}
\begin{align}\label{eq:critic_loss}
\mathcal{L}(\phi)=\frac{1}{B}\sum_{i=1}^{B}\Big(\mathcal{Q}_{\phi}\big(s^{(i)},a^{(i)},\zeta_{\text{pert}}^{(i)}\big) - y^{(i)}\Big)^2.
\end{align}
Since $\zeta_{\text{pert}}^{(i)}$ is perturbed locally around $\zeta_k$ for every collected data point, the critic observes variation in its current-step environment parameter copy input, which permits $\nabla_{\zeta}\mathcal{Q}_{\phi}(s,a,\zeta)\big|_{\zeta=\zeta_k}$ to serve as an estimate of $\nabla_{\zeta}\mathcal{Q}^{\mathrm{H}}_{\pi_{\theta_k},\zeta_k}(s,a,\zeta)\big|_{\zeta=\zeta_k}$ --- the first term of {\color{black}$\hat{g}_{\pi_{\theta_k}}^{\zeta_k}$ in (\ref{eq:grad_est_zeta})}. Further, $\mathcal{Q}_{\phi}(s,a,\zeta_k)$ estimates $\mathcal{Q}^{\mathrm{H}}_{\pi_{\theta_k},\zeta_k}(s,a,\zeta_k)$, required in the second term {\color{black}of $\hat{g}_{\pi_{\theta_k}}^{\zeta_k}$ and in $\hat{h}_{\pi_{\theta_k}}^{\zeta_k}$ \eqref{eq:grad_est_theta}}. In our proposed model-free algorithm, $\alpha$ is gradually annealed to zero. As a consequence, the entropy-regularized  $J_{\mathrm{H}}(\pi,\zeta)$ approaches the original objective $J(\pi,\zeta)$ in \eqref{eq:objective}, and we ultimately solve the co-design problem in \eqref{eq:objective}. See Algorithm \ref{alg:envpg_ac} for details.

\section{Simulation}\label{sec:simulation}
We demonstrate Algorithm~\ref{alg:envpg_ac} on the {\color{black}aerial relay} network design problem of Fig.~\ref{fig:UAV_Network}. The network consists of $271$ fixed user nodes $\{n_i\}$ located at $\{x_i\in\mathbb{R}^2\}$ and a destination node $\delta$ located at $z\in\mathbb{R}^2$. Five {\color{black}aerial relay nodes} $\{f_j\}$ serve as intermediate hops; each message from a user node must reach $\delta$ through a sequence of hops via the relays. The cost of a hop is the squared Euclidean distance between the two nodes, approximating the communication delay{\color{black}, plus a fixed per-hop overhead of $0.2$ independent of the distance}. The relay locations $\{y_j\}_{j=1}^{5}$ are stacked into the environment parameter $\zeta\in\mathbb{R}^{10}$, and the objective is to determine $\zeta$ jointly with the routing policy to minimize the total communication cost.

{\color{black}\textit{MDP:} A state is the node at which a message currently resides,} so $\mathcal{S}=\{\{n_i\},\{f_j\},\delta\}$ with $|\mathcal{S}|=277$, where $\delta$ is the cost-free terminal state. The action space $\mathcal{A}=\{\{f_j\},\delta\}$ has $|\mathcal{A}|=6$ actions, each recommending the next hop. From a user node, a message can only be sent to the relays; from a relay $f_j$, to any other relay $f_{j'}\neq f_j$ or to $\delta$. {\color{black}We write $\mathbf{x}_s\in\mathbb{R}^2$ for the location of node $s$, that is, $\mathbf{x}_s=x_i$ if $s=n_i$, $\mathbf{x}_s=y_j$ if $s=f_j$, and $\mathbf{x}_s=z$ if $s=\delta$, where $y_j$ is the location of relay $f_j$ and $\zeta=[y_1^\top,\dots,y_5^\top]^\top$.} The transition kernel is parameterized by $\kappa\in(0,1]$:
\begin{align}\label{eq:transition_kernel}
p_\zeta(s'|s,a) = \begin{cases}
\kappa & \text{if } s' = a \\
(1-\kappa)\, q_\zeta(s'|s,a) & \text{otherwise}
\end{cases},
\end{align}
{\color{black}i.e., the recommended hop is realized with probability $\kappa$, and otherwise the message reaches one of the other valid next hops, where $q_\zeta(s'|s,a)\propto\exp\{-\eta\|\mathbf{x}_{s'}-\mathbf{x}_s\|^2\}$ over $\mathcal{A}(s)\setminus\{a\}$ makes nearer nodes more likely ($\eta=100$). The transition kernel is a simplified model accounting for communication link failures and rerouting.} We vary $\kappa\in\{1.00,0.95,0.90,0.85\}$. The cost is $c_\zeta(s,a,s')=\|\mathbf{x}_{s'}-\mathbf{x}_s\|^2{\color{black}+0.2}$ {\color{black}for every hop out of a non-terminal state}, the initial distribution $\mu$ is uniform over the user nodes, and $\gamma=1$.

The kernel and cost above define an {\em emulator} with which Algorithm~\ref{alg:envpg_ac} interacts: they are never used explicitly, and all gradients are estimated from sampled trajectories, so the setup is model-free. The actor and critic are two-layer neural networks whose inputs are the state and action coordinates and their one-hot encodings, and the temperature $\alpha$ is annealed geometrically to $\approx 0$ so that $J_{\mathrm{H}}(\pi,\zeta)\rightarrow J(\pi,\zeta)$.

\textit{Benchmarks:} We compare Algorithm~\ref{alg:envpg_ac} with two benchmarks. The first is a {\em model-based oracle} with full access to $p_\zeta$ and $c_\zeta$: at every iteration it computes $J_{\mathrm{H}}(\pi^*,\zeta)$ at the values obtained by perturbing each coordinate of $\zeta_k$ by $\pm\epsilon$, estimates $\nabla_\zeta J_{\mathrm{H}}(\pi^*,\zeta_k)$ by central differences, and updates $\zeta_k$ by gradient descent, with $\alpha$ annealed as in Algorithm~\ref{alg:envpg_ac}. {\color{black}The second is a score-function (SF) co-design baseline, which samples each relay locations from a distinct Gaussian design distribution and updates its mean and variance with the corresponding score-function gradients \cite{schaff2019jointly, cauz2024reinforcement}.}

\textit{Results{\color{black}\footnote{{\color{black}Source code to reproduce experiments: \url{https://github.com/asrivastava93/paper-envpg-implementation}}}}:} {\color{black}The descent steps in Algorithm~\ref{alg:envpg_ac} are implemented using Adam \cite{kingma2014adam}.} Fig.~\ref{fig:network_designs}(a) shows the final relay placement and routing learned by Algorithm~\ref{alg:envpg_ac}; user nodes are colored by the relay they are routed to under the learned policy, and the arrows are the recommended actions from the relays. {\color{black}Fig.~\ref{fig:network_designs}(b) shows the relay positions every $500$ updates, from the initial placement (circles) to the final one (triangles). The paths are not straight, since the entropy-regularized cost changes as $\alpha$ is annealed, and the relays settle within about $2{,}000$ updates.} Table~\ref{tab:jstar} reports the final objective $J^*$, averaged over $6$ runs with different initializations of $\zeta$ at each $\kappa$. Algorithm~\ref{alg:envpg_ac} {\color{black}is within $1.0\%$ (on average) of the oracle and attains a $J^*$ that is on average $3.4\%$ lower than that of SF, while optimizing only half as many design variables}.

\begin{table}[t]
{\color{black}
\setlength{\tabcolsep}{3pt}
\centering
\caption{Final objective $J^*$ across $\kappa$ (mean $\pm$ standard deviation).}
\vspace{-0.3cm}
\label{tab:jstar}
\begin{tabular*}{\columnwidth}{@{\extracolsep{\fill}}l cccc@{}}
\toprule
Method & $\kappa=1.00$ & $\kappa=0.95$ & $\kappa=0.90$ & $\kappa=0.85$ \\
\midrule
EnvPG & $14.36\pm0.45$ & $14.81\pm0.03$ & $15.71\pm0.10$ & $16.42\pm0.01$ \\
SF \cite{schaff2019jointly},\cite{cauz2024reinforcement} & $14.75\pm0.66$ & $15.30\pm0.25$ & $16.32\pm0.21$ & $17.09\pm0.38$ \\
Oracle & $13.92\pm 0.25$ & $14.74\pm 0.001$ & $15.74\pm 0.10$ & $16.40\pm 0.18$ \\
\bottomrule
\end{tabular*}
\vspace{-0.5cm}}
\end{table}

{\color{black}To validate the gradient estimate directly, we compare it with the exact gradient of $J_{\mathrm{H}}(\pi,\zeta)$, computed from the model at the same $(\pi,\zeta)$. Fig.~\ref{fig:network_designs}(c) plots the cosine similarity between the two for the stochastic estimate observed during training (from $51$ sampled trajectories) and for a stochastic gradient estimated separately from $N=5{,}000$ trajectories (median over the $24$ runs; shaded: interquartile range). Fig.~\ref{fig:network_designs}(d) plots $J^*(\zeta_k)$, the exact optimal cost at the current placement, which decreases and plateaus for every $\kappa$, so the design settles at the values in Table~\ref{tab:jstar}. The cosine is high while $J^*$ decreases rapidly, and its median remains positive throughout, so the stochastic estimate remains a descent direction. Once $J^*$ has converged (after about $2{,}000$ updates), the placement no longer changes appreciably and the cosine is no longer informative.} Finally, Fig.~\ref{fig:network_designs}(e) plots the norm of the stochastic gradient estimate (bold: median over initializations of its moving average over $200$ updates; shaded: range over initializations); it decreases and plateaus at a small residual. Since $\alpha\rightarrow 0$, $J_{\mathrm{H}}(\pi,\zeta)\rightarrow J(\pi,\zeta)$ in~\eqref{eq:objective}, recovering the original co-design objective.

\section{Conclusion and Future Work}\label{sec:FutureWork}
We established the EnvPG theorem, a gradient of the RL objective with respect to the environment design parameters. Building on it, we propose a model-free algorithm that jointly optimizes the policy and the environment. Several directions remain open. {\color{black} On the theoretical side, characterizing the approximators of $\mathcal{Q}_{\pi,\xi}(s,a,\zeta)$ that yield unbiased estimates of $\nabla_\theta J(\pi_\theta,\zeta)$ and $\nabla_\zeta J(\pi_\theta,\zeta)$ would be a natural step toward a formal convergence analysis of Algorithm~\ref{alg:envpg_ac}. Finally, extending the framework to constrained MDPs, and to physical platforms, where reconfiguration dynamics, costs and safety constraints must be taken into account, would broaden its applicability to wider engineering problems.}

\bibliographystyle{IEEEtran}
\bibliography{IEEEabrv}

@book{Sutton1998,
  author = {Sutton, Richard S. and Barto, Andrew G.},
  edition = {Second},
  publisher = {The MIT Press},
  title = {Reinforcement Learning: An Introduction},
  url = {http://incompleteideas.net/book/the-book-2nd.html},
  year = {2018 }
}

@book{bertsekas1996neuro,
  title={Neuro-Dynamic Programming},
  author={Bertsekas, Dimitri P and Tsitsiklis, John N},
  year={1996},
  publisher={Athena Scientific},
  address={Belmont, MA}
}

@InProceedings{pmlr-v100-luck20a,
  title = 	 {Data-efficient Co-Adaptation of Morphology and Behaviour with Deep Reinforcement Learning},
  author =       {Luck, Kevin Sebastian and Amor, Heni Ben and Calandra, Roberto},
  booktitle = 	 {Proceedings of the Conference on Robot Learning},
  pages = 	 {854--869},
  year = 	 {2020},
  editor = 	 {Kaelbling, Leslie Pack and Kragic, Danica and Sugiura, Komei},
  volume = 	 {100},
  series = 	 {Proceedings of Machine Learning Research},
  month = 	 {30 Oct--01 Nov},
  publisher =    {PMLR},
  url = 	 {https://proceedings.mlr.press/v100/luck20a.html}
}

@ARTICLE{9517030,
  author={Srivastava, Amber and Salapaka, Srinivasa M.},
  journal={IEEE Transactions on Cybernetics}, 
  title={Parameterized MDPs and Reinforcement Learning Problems—A Maximum Entropy Principle-Based Framework}, 
  year={2022},
  volume={52},
  number={9},
  pages={9339-9351},
  doi={10.1109/TCYB.2021.3102510}}

@article{williams1992simple,
  title={Simple statistical gradient-following algorithms for connectionist reinforcement learning},
  author={Williams, Ronald J},
  journal={Machine learning},
  volume={8},
  number={3},
  pages={229--256},
  year={1992},
  publisher={Springer}
}

@inproceedings{he2024morph,
  title={Morph: Design co-optimization with reinforcement learning via a differentiable hardware model proxy},
  author={He, Zhanpeng and Ciocarlie, Matei},
  booktitle={2024 IEEE International Conference on Robotics and Automation (ICRA)},
  pages={7764--7771},
  year={2024},
  organization={IEEE}
}

@article{bolland2022jointly,
  title={Jointly Learning Environments and Control Policies with Projected Stochastic Gradient Ascent},
  author={Bolland, Adrien and Boukas, Ioannis and Berger, Mathias and Ernst, Damien},
  journal={Journal of Artificial Intelligence Research},
  volume={73},
  pages={117--171},
  year={2022}
}

@inproceedings{jackson2021orchid,
  title={Orchid: optimisation of robotic control and hardware in design using reinforcement learning},
  author={Jackson, Lucy and Walters, Celyn and Eckersley, Steve and Senior, Pete and Hadfield, Simon},
  booktitle={2021 IEEE/RSJ International Conference on Intelligent Robots and Systems (IROS)},
  pages={4911--4917},
  year={2021},
  organization={IEEE}
}

@inproceedings{haarnoja2018soft,
  title={Soft actor-critic: Off-policy maximum entropy deep reinforcement learning with a stochastic actor},
  author={Haarnoja, Tuomas and Zhou, Aurick and Abbeel, Pieter and Levine, Sergey},
  booktitle={International conference on machine learning},
  pages={1861--1870},
  year={2018},
  organization={Pmlr}
}

@article{bhatia2021evolution,
  title={Evolution gym: A large-scale benchmark for evolving soft robots},
  author={Bhatia, Jagdeep and Jackson, Holly and Tian, Yunsheng and Xu, Jie and Matusik, Wojciech},
  journal={Advances in Neural Information Processing Systems},
  volume={34},
  pages={2201--2214},
  year={2021}
}

@article{ha2019reinforcement,
  title={Reinforcement learning for improving agent design},
  author={Ha, David},
  journal={Artificial life},
  volume={25},
  number={4},
  pages={352--365},
  year={2019},
  publisher={MIT Press One Rogers Street, Cambridge, MA 02142-1209, USA journals-info~…}
}

@article{pao2024control,
  title={Control co-design of wind turbines},
  author={Pao, Lucy Y and Pusch, Manuel and Zalkind, Daniel S},
  journal={Annual Review of Control, Robotics, and Autonomous Systems},
  volume={7},
  year={2024},
  publisher={Annual Reviews}
}

@article{kingma2014adam,
  title={Adam: A method for stochastic optimization},
  author={Kingma, Diederik P and Ba, Jimmy},
  journal={arXiv preprint arXiv:1412.6980},
  year={2014}
}

@inproceedings{schaff2019jointly,
  title={Jointly learning to construct and control agents using deep reinforcement learning},
  author={Schaff, Charles and Yunis, David and Chakrabarti, Ayan and Walter, Matthew R},
  booktitle={2019 international conference on robotics and automation (ICRA)},
  pages={9798--9805},
  year={2019},
  organization={IEEE}
}

@article{sutton1999policy,
  title={Policy gradient methods for reinforcement learning with function approximation},
  author={Sutton, Richard S and McAllester, David and Singh, Satinder and Mansour, Yishay},
  journal={Advances in neural information processing systems},
  volume={12},
  year={1999}
}

@article{mnih2015human,
  title={Human-level control through deep reinforcement learning},
  author={Mnih, Volodymyr and Kavukcuoglu, Koray and Silver, David and Rusu, Andrei A and Veness, Joel and Bellemare, Marc G and Graves, Alex and Riedmiller, Martin and Fidjeland, Andreas K and Ostrovski, Georg and others},
  journal={nature},
  volume={518},
  number={7540},
  pages={529--533},
  year={2015},
  publisher={Nature Publishing Group}
}

@InProceedings{silver2014deterministic,
  title = 	 {Deterministic Policy Gradient Algorithms},
  author = 	 {Silver, David and Lever, Guy and Heess, Nicolas and Degris, Thomas and Wierstra, Daan and Riedmiller, Martin},
  booktitle = 	 {Proceedings of the 31st International Conference on Machine Learning},
  pages = 	 {387--395},
  year = 	 {2014},
  editor = 	 {Xing, Eric P. and Jebara, Tony},
  volume = 	 {32},
  number =       {1},
  series = 	 {Proceedings of Machine Learning Research},
  address = 	 {Bejing, China},
  month = 	 {22--24 Jun},
  publisher =    {PMLR},
  url = 	 {https://proceedings.mlr.press/v32/silver14.html}
}

@inproceedings{
cauz2024reinforcement,
title={Reinforcement Learning for Efficient Design and Control Co-optimisation of Energy Systems},
author={Marine Cauz and Adrien Bolland and Christophe Ballif and Nicolas Wyrsch},
booktitle={ICML 2024 AI for Science Workshop},
year={2024}
}

@article{garcia2019control,
  title={Control Co-Design: an engineering game changer},
  author={Garcia-Sanz, Mario},
  journal={Advanced Control for Applications: Engineering and Industrial Systems},
  volume={1},
  number={1},
  pages={e18},
  year={2019},
  publisher={Wiley Online Library}
}

@conference{Fox2016Glearning,
  Title = "Taming the Noise in Reinforcement Learning via Soft Updates",
  Author = "Roy Fox and Ari Pakman and Naftali Tishby",
  Booktitle = "32nd Conference on Uncertainty in Artificial Intelligence (UAI)",
  Year = "2016"
}

\end{document}